\documentclass[envcountsect]{llncs}
\usepackage[utf8]{inputenc}
\usepackage{amsmath}
\usepackage{graphicx}
\usepackage{epstopdf}
\usepackage{cases}
\usepackage[obeyspaces]{url}
\PassOptionsToPackage{obeyspaces}{url}
\usepackage{xparse}

\newcommand{\alignc}[1]{ {\normalsize \begin{align*}#1\end{align*}} }
\newcommand{\equationc}[1]{ {\normalsize \begin{equation}#1\end{equation}} }
\newcommand{\smallcenter}[1]{ {\normalsize \[#1\]} }
\newcommand\Bceil[1]{\Bigl\lceil#1\Bigr\rceil}
\newcommand\Bfloor[1]{\Bigl\lfloor#1\Bigr\rfloor}
\newcommand\norm[1]{\lvert\lvert#1\rvert\rvert}
\newcommand\epi{\bar{\pi}}

\newcommand\eiota{\bar{\iota}}
\newcommand\spi{\eiota {\epi}^{-1}}

\ExplSyntaxOn
\NewDocumentCommand{\cycle}{ O{\;} m }
{
	(
	\alec_cycle:nn { #1 } { #2 }
	)
}

\seq_new:N \l_alec_cycle_seq
\cs_new_protected:Npn \alec_cycle:nn #1 #2
{
	\seq_set_split:Nnn \l_alec_cycle_seq { , } { #2 }
	\seq_use:Nn \l_alec_cycle_seq { #1 }
}
\ExplSyntaxOff

\ExplSyntaxOn
\NewDocumentCommand{\cyclen}{ O{} m }
{
	(
	\alec_cyclen:nn { #1 } { #2 }
	)
}

\seq_new:N \l_alec_cyclen_seq
\cs_new_protected:Npn \alec_cyclen:nn #1 #2
{
	\seq_set_split:Nnn \l_alec_cyclen_seq { , } { #2 }
	\seq_use:Nn \l_alec_cyclen_seq { #1 }
}
\ExplSyntaxOff

\begin{document}
	\title{A barrier for further approximating Sorting By Transpositions}
	
	\author{Luiz Augusto G. da Silva\inst{1} \and
		Luis Antonio B. Kowada\inst{2} \and
		Maria Emília M. T. Walter\inst{1}}
	
	\authorrunning{Silva et al.}
	
	\institute{Departamento de Ciência da Computação, Universidade de Brasília, Brazil \email{laugustogarcia@gmail.com,mariaemilia@unb.br} \and
		Instituto de Computação, Universidade Federal Fluminense, Niterói, Brazil \email{luis@ic.uff.br}}
	\maketitle
	
	\begin{abstract}
		The \textsc{Transposition Distance Problem} (TDP) is a classical problem in genome rearrangements which seeks to determine the minimum number of transpositions needed to transform a linear chromosome into another represented by the permutations $\pi$ and $\sigma$, respectively. This paper focuses on the equivalent problem of \textsc{Sorting By Transpositions} (SBT), where $\sigma$ is the identity permutation $\iota$. Specifically, we investigate palisades, a family of permutations that are ``hard'' to sort, as they require numerous transpositions above the celebrated lower bound devised by Bafna and Pevzner. By determining the transposition distance of palisades, we were able to provide the exact transposition diameter for $3$-permutations (TD3), a special subset of the Symmetric Group $S_n$, essential for the study of approximate solutions for SBT using the simplification technique. The exact value for TD3 has remained unknown since Elias and Hartman showed an upper bound for it. Another consequence of determining the transposition distance of palisades is that, using as lower bound the one by Bafna and Pevzner, it is impossible to guarantee approximation ratios lower than $1.375$ when approximating SBT. This finding has significant implications for the study of SBT, as this problem has been subject of intense research efforts for the past 25 years.
	\end{abstract}
	
	\section{Introduction}
	
	Mutational events involving large portions of DNA in a chromosome are studied in the area of genome rearrangements. These are rare events and their study is crucial for understanding the evolutionary history of organisms~\cite{koonin2005orthologs,NadeauTaylor1984,PalmerHerbon1988,YuePhylogenetics}. One of these events is \emph{transposition}, which informally corresponds to the exchange of positions of two adjacent blocks of genes in a linear chromosome. Transposition can also be seen as cutting out a gene block and then pasting it elsewhere (in the same chromosome). Then the question arises: what is the minimum number $t$ of transpositions required to transform one linear chromosome, represented by a permutation $\pi$, into another one represented by $\sigma$, both consisting of the same set of genes? This classical problem is referred to as \textsc{Transposition Distance Problem} (TDP) and, if we consider the target chromosome $\sigma$ as being the identity $\iota$, then TDP turns into the equivalent problem of \textsc{Sorting By Transpositions} (SBT). In this case, the number $t$ is called the \emph{transposition distance} of $\pi$ and is denoted by $d_t(\pi)$.
	
	SBT is a $\mathcal{NP}$-hard problem~\cite{bulteau2012sorting} and the first approximate algorithm for solving it was devised by Bafna and Pevzner~\cite{BafnaPevzner1998}, with an approximation ratio of $1.5$. Later, Elias and Hartman~\cite{EliasHartman2006} devised an $1.375$-approximation algorithm running in quadratic time, using a technique called \emph{simplification}, which transforms an arbitrary permutation provided as input into a new one, seemingly easier to handle with. In 2022, Silva et al.~\cite{silva2022new} showed that a side effect of simplification could make the algorithm of Elias and Hartman~\cite{EliasHartman2006} need one extra transposition above the approximation ratio of $1.375$ and so they devised a new $1.375$-approximation algorithm, not using simplification, but with a high time complexity of $O(n^6)$. Recently, Alexandrino et al.~\cite{alexandrino20221} published a result making it possible to approximate SBT with a $1.375$ ratio, without extra transpositions, in $O(n^5)$ time.
	
	In this work, we consider the palisades, a family of permutations which are ``hard'' to sort, in the sense that, to be sorted, numerous transpositions above the well-known lower bound established by Bafna and Pevzner~\cite{BafnaPevzner1998} are required. By determining the transposition distance of palisades, we are able to determine the exact value for the transposition diameter for $3$-permutations (denoted by $TD3$), a special subset of the Symmetric Group $S_n$, essential for the study of approximate solutions for SBT using the simplification technique~\cite{Hartman2005,hartman2006simpler}. This problem has been open since Elias and Hartman~\cite{EliasHartman2006} have shown an upper bound for it. Another surprising consequence of the result on the transposition distance of palisades, as we will show, is that, using as lower bound, the one by Bafna and Pevzner~\cite{BafnaPevzner1998}, it is impossible to devise approximate solutions for SBT with approximation ratios lower than $1.375$.
	
	We organise the text as follows. First, we briefly review the basic concepts of permutation groups needed to understand the algebraic formalism used in this work. Next, we will review SBT using the algebraic approach introduced by Silva et al.~\cite{silva2022new}. Subsequently, we will present our findings regarding palisades and, lastly, we include a brief discussion on the results and conclude. A very brief overview on the cycle graph~\cite{BafnaPevzner1998} is given in Appendix~\ref{appendix-a}.
	
	\section{Permutation groups}
	
	A permutation group is a group of bijections from a set to itself (i.e., permutations), whose group operation is the usual function composition. The Symmetric Group $S_n$ is the permutation group consisting of all permutations on a finite set $E$ containing $n$ symbols. Let $\alpha$ and $\beta$ be two permutations in $S_n$. The product $\alpha\beta$ is obtained by composing the two permutation as functions. That is, the result of $\alpha\beta$ is the permutation in which each element $x$ of $E$ is mapped to $\alpha(\beta(x))$.
	
	A \emph{cycle} $\gamma$ is a sequence $(c_1\;c_2\dots c_\kappa)$ of symbols of $E$ such that $\alpha(c_1)=c_2$, $\alpha(c_2) = c_3$, \dots, $\alpha(c_{\kappa-1})=c_\kappa$, $\alpha(c_\kappa)=c_1$. In other words, the symbols in $\gamma$ are cyclically permuted by $\alpha$. Due to its cyclic nature, $\gamma$ can be written starting by any of the symbols moved by it. If two cycles have no common symbols, then they commute, meaning that the order in which we compose them does not matter. In this case, these two cycles are said to be \emph{disjoint}. A \emph{disjoint cycle decomposition} of $\alpha$ is a (unique) representation of $\alpha$ as a product of disjoint cycles. Additionally, the number $\kappa$ is the \emph{length} of $\gamma$ or, alternatively, $\gamma$ is a $\kappa$-cycle. If $x \in E$ is an element such that $\alpha(x)=x$, then $x$ is said to be a \emph{fixed element}. Fixed elements are sometimes omitted when representing a permutation as a product of cycles, but when necessary, a fixed element $x$ is represented by the $1$-cycle $(x)$. The identity permutation $\iota$ is the permutation which fixes all $E$ elements.
	
	\begin{example}
		The permutation $(1\;2\;3)(5\;7) \in S_7$ is the permutation which sends $1$ to $2$, $2$ to $3$, $3$ to $1$, $5$ to $7$ and $7$ to $5$, and leaves all the other symbols fixed.
	\end{example}
	
	It is well known that any permutation on $S_n$ can be written as a (not unique) product of $2$-cycles. Furthermore, if a permutation can be written as a product of an even (odd) number of $2$-cycles, then it is impossible to write the same permutation using a product of an odd (even) number of $2$-cycles. Proofs for these results can be found in algebra textbooks~\cite{Gallian2009,dummit2004abstract}.
	
	\begin{example}
		The permutation $(1\;3\;2)(4\;5\;7)$ can be written as a product of four $2$-cycle as $(1\;3)(3\;2)(4\;5)(5\;7)$. But this same permutation can be written as a product of $2$-cycles in infinitely many other ways.
	\end{example}
	
	Lastly, we highlight that many results in permutation groups consider the parity of the permutations, defined as follows. Suppose we can write $\alpha$ as a product of $p$ $2$-cycles. In this case, $\alpha$ is an \emph{even (odd)} permutation if $p$ is even (odd). An important result concerning the parity of permutations is that, assuming that $\alpha$ and $\beta$ have the same parity, then the product $\alpha\beta$ is even~\cite{dummit2004abstract}.
	
	\section{Preliminaries}
	
	The order in which a set of $n$ genes appear along a linear chromosome, in the context of genome rearrangements, is usually represented by a permutation. For instance, given a set of $n$ genes labelled $1$ through $n$, a permutation of these genes can be represented as a sequence of integers $\pi=[\pi_1\;\pi_2\dots \pi_n]$, where each $\pi_i$, $i \in [1,n]$, indicates the position of gene $i$ in the chromosome. In our approach, the chromosome $\pi$ is represented by the $(n+1)$-cycle $\epi=(0\;\pi_1\;\pi_2\dots \pi_n)$. In its turn, the identity permutation $\iota$ is represented by $\eiota=(0\;1\;2\dots n)$. A transposition exchanging the positions of two adjacent blocks of genes is modelled by the multiplication of a $3$-cycle $\tau=(\pi_i\;\pi_j\;\pi_k)$ with $\epi=(\pi_0\;\pi_1\dots \pi_{i-1}\;\pi_i\;\pi_{i+1}\dots \pi_{j-1}\;\pi_j\;\pi_{j+1}\dots \pi_{k-1}\;\pi_k\dots \pi_n)$, resulting the $(n+1)$-cycle $(\pi_0\;\pi_1\dots \pi_{i-1}\;\pi_j\;\pi_{j+1}\dots \pi_{k-1}\;\pi_i\;\pi_{i+1}\dots \pi_{j-1}\;\pi_k\dots \pi_n)$. When performing this multiplication, we can say that we are ``applying'' $\tau$ on $\epi$ and, for this reason, $\tau$ is \emph{applicable} (to $\epi$). However, we remark that there is a condition for a $3$-cycle $\tau=(\pi_i\;\pi_j\;\pi_k)$ to be applicable, which is that the symbols $\pi_i$, $\pi_j$, and $\pi_k$ must be in the same cyclic order in both $\tau$ and $\epi$ (as in the definition above). Otherwise, the product $\tau\epi$ will not be an $(n+1)$-cycle and consequently, will not model a chromosome.
	
	\begin{example}
		Let $\epi=(0\;4\;5\;2\;3\;7\;6\;1\;8)$. Applying the $3$-cycle $\tau=(2\;6\;8)$ to $\epi$ simulates a transposition as $\tau\epi=(0\;4\;5\;6\;1\;2\;3\;7\;8)$. However, the $3$-cycle $\tau'=(2\;4\;7)$ cannot be applied to $\epi$. This is because $\tau'\epi$ does not correspond to a transposition in our approach, as $\tau'\epi=(0\;7\;8)(1\;4\;5\;6)(2\;3)$, which is not a $(n+1)$-cycle.
	\end{example}
	
	The problem of \textsc{Sorting By Transpositions} is therefore the problem of, given a $(n+1)$-cycle $\epi$, finding the minimum $t$, denoted by $d_t(\epi)$, such that
	\equationc{
		\tau_t\dots\tau_2\tau_1\epi=\eiota\label{eqtdp},
	}
	where $\tau_1$ is a $3$-cycle applicable to $\epi$ and each $\tau_i$, $i\in[2,t]$, is a $3$-cycle  applicable to $\tau_{i-1}\dots\tau_2\tau_1\epi$.
	
	\subsection{Lower bound for transposition distance}
	
	Mira and Meidanis~\cite{mira2005algebraic} devised a measure for an even permutation $\alpha \in S_n$ called $3$-norm, denoted as $\norm{\alpha}3$. The $3$-norm of $\alpha$ is defined as the minimum $\ell$ such that $\beta_{\ell}\dots\beta_2\beta_1=\alpha$, where each $\beta_i$, $i \in [1,\ell]$, is a $3$-cycle. Let ${c^\circ}_{odd}(\alpha)$ denote the number of odd-length cycles (including $1$-cycles) in $\alpha$. Based on this, Mira and Meidanis~\cite{mira2005algebraic} proved the following lemma:
	
	\begin{lemma}[Mira and  Meidanis~\cite{mira2005algebraic}]\label{lem:norm}
		$\norm{\alpha}_3 = \frac{n-{c^\circ}_{odd}(\alpha)}{2}$.
	\end{lemma}
	
	Now let us exam Equation~\ref{eqtdp}. It can be rewritten as:
	
	\equationc{
		\tau_t\dots\tau_2\tau_1=\eiota\epi^{-1}\label{eqtdp2}.
	}
	
	This means that the product of the $3$-cycles that transform $\epi$ into $\eiota$ is equal to the product $\eiota\epi^{-1}$. Furthermore, since $\epi$ and $\eiota$ are both $(n+1)$-cycles, the product $\spi$ is always even. As a result, we can state a lower bound for SBT as follows:
	
	\begin{lemma}[Mira and Meidanis~\cite{mira2005algebraic}]\label{lem:lb-norm}
		\alignc{d_t(\epi) & \geq \norm{\eiota\epi^{-1}}3\\ & \geq \frac{n+1-{c^\circ}_{odd}(\eiota\epi^{-1})}{2}.}
	\end{lemma}
	
	Observe that the maximum number of cycles in $\spi$ is achieved if and only if $\spi$ is the identity permutation $\iota$. In this case, $\iota$ has $n+1$ cycles, all of which are odd-length (notably, they are all of length $1$).
	
	Lastly, Silva et al.\cite{silva2022new} have shown that there is a one-to-one mapping from $\spi$ to the cycle graph~\cite{BafnaPevzner1998} of $\pi$. Therefore, the lower bound stated above is equivalent to the well-known one devised by Bafna and Pevzner~\cite{BafnaPevzner1998}, which is formalised in Appendix~\ref{appendix-a} (Theorem~\ref{th:bafna}).
	
    \subsection{Cycles of $\spi$}

    While the concepts in this section apply to any $(n+1)$-cycle $\epi$, as previously defined by Silva et al. \cite{silva2022new}, in the specific context of this paper $\epi$ will typically be a $3$-permutation. A $(n+1)$-cycle $\epi$, $(n+1)\equiv 3\pmod 0$, is a $3$-permutation when $\spi$ consists only of $1$- or $3$-cycles.
 
	Let $\gamma=(a\dots b\dots c\dots)$ be one of the disjoint cycles of $\spi$. If $\epi^{-1}=(a\dots c\dots b$ $\dots)$, meaning that the symbols $a$, $b$, and $c$ appear in $\gamma$ in a cyclic order different from the one in $\epi^{-1}$, then we define $\gamma$ as an \emph{oriented cycle}. On the other hand, if there is no oriented triplet in $\gamma$, then $\gamma$ is called an \emph{unoriented cycle}.
	
	Let $\delta=(a\;b\dots)$ and $\epsilon=(d\;e\dots)$ be two cycles of $\spi$. If $\epi^{-1}=(a\dots e\dots $ $b\dots d\dots)$, meaning that the symbols of pairs $(a,b)$ and $(d,e)$ appear in an alternating order in $\epi^{-1}$, we say that $\delta$ and $\epsilon$ are \emph{intersecting} cycles. A special case is when $\delta=(a\;b\;c\dots)$ and $\epsilon=(d\;e\;f\dots)$ are such that $\epi^{-1}=(a\dots e\dots b\dots f\dots c\dots d\dots)$, implying that the symbols of triplets $(a,b,c)$ and $(d,e,f)$ appear in an alternating order in $\epi^{-1}$. Under this circumstance, $\delta$ and $\epsilon$ are considered to be \emph{interleaving} cycles.
	
	\begin{example}
		Let $\epi=(0\;7\;6\;1\;9\;5\;8\;4\;3\;2)$ and $\spi=(0\;3\;5)(1\;7)(2\;4\;9)(6\;8)$. The cycles $(1\;7)$ and $(6\;8)$ are examples of intersecting cycles, while $(0\;3\;5)$ and $(2\;4\;9)$ are interleaving cycles.
	\end{example}
	
	From Equation~\ref{eqtdp2}, we have that $\spi{\tau_1}^{-1}\dots{\tau_t}^{-1}=\iota$, meaning that the application of transpositions $\tau_1$, $\dots$, $\tau_t$ sorting $\epi$ (i.e., transforming $\epi$ into $\eiota$) can be seen as the incremental multiplication of $\spi$ by ${\tau_1}^{-1}$, $\dots$, ${\tau_t}^{-1}$. Let us denote the difference ${c^\circ}_{odd}(\spi\tau^{-1})-{c^\circ}_{odd}(\spi)$ as $\Delta {c^\circ}_{odd}(\spi, \tau)$.
	
	\begin{proposition}[Meidanis, Dias, and Mira~\cite{MeidanisDias2000,mira2005algebraic}]
		If $\tau$ is an applicable $3$-cycle, then $\Delta {c^\circ}_{odd}(\spi, \tau)\in\{-2,0,2\}$.\label{prop}
	\end{proposition}
	
	An applicable $3$-cycle $\tau$ such that $\Delta {c^\circ}_{odd}(\spi, \tau)=\mu$ is referred to as a $\mu$-move. Proposition~\ref{prop} identifies three types of moves, namely $(-2)$-, $0$- and $2$-move. An $(x,y)$-sequence, where $x\geq y$, denotes a sequence of $x$ applicable $3$-cycles $\tau_1$, $\dots$, $\tau_x$ with at least $y$ of them being $2$-moves. If $x \leq a$ and $\frac{x}{y} \leq \frac{a}{b}$, then the $(x,y)$-sequence is called a $\frac{a}{b}$-sequence.
	
	A \emph{configuration} $\Gamma$ is a disjoint product of cycles of $\spi$. We say $\Gamma$ is \emph{connected} if, for any two cycles $\gamma_1$ and $\gamma_m$ of $\Gamma$, there are cycles $\gamma_2,\dots,\gamma_{m-1}$ in $\Gamma$ such that for each $i \in [1, m-1]$, $\gamma_i$ intersects or interleaves with $\gamma_{i+1}$. $\Gamma$ is considered a \emph{component} if it consists of only one oriented cycle not intersecting or interleaving with any other cycle of $\spi$; or it consists of a maximal connected configuration of $\spi$. Lastly, we define an \emph{unoriented interleaving pair} as a product $\gamma\delta$ of two $3$-cycles such that $\gamma\delta$ is a component of $\spi$ and both $\gamma$ and $\delta$ are unoriented and interleave with each other.
	
	\begin{example}
		Consider $\epi=(0\;3\;2\;1\;4\;6\;5)$ with $\spi=(0\;6\;5)(1\;3)(2\;4)$. $(0\;6\;5)$ and $(1\;3)(2\;4)$ are the two components of $\spi$.
	\end{example}
	
	\begin{example}
		Let $\epi=(0\;5\;4\;3\;8\;7\;6\;2\;1\;9\;14\;13\;12\;11\;10)$ with $\spi=(0\;11\;13)(1\;3\;5)(2\;7\;9)(4\;6\;8)(10\;12 \;14)$. The products $(0\;11\;13)(10\;12 \;14)$ and $(1\;3\;5)(2\;7\;9)(4\;6\;8)$ are both components of $\spi$. In particular, $(0\;11\;13)(10\;12$ $\;14)$ is an unoriented interleaving pair.
	\end{example}

	Finally, let $\epi$ be a $(n+1)$-cycle. We say $u \geq 2$ components are \emph{cleared} when we are able to apply a $\frac{3}{2}$-sequence $\tau_1,\dots,\tau_{s}$ on $\epi$ so that $\norm{\spi\tau_1^{-1}\dots\tau_{s}^{-1}}$ $=\norm{\spi}-2u$.
	
	\section{Palisades and their transposition distance}
	
	This section introduces the concept of palisades, a family of permutations that are ``hard'' to sort, as they require a significant number of transpositions above the lower bound by Bafna and Pevzner (Theorem~\ref{th:bafna}) to be sorted. In this regard, this section aims to define their transposition distance.
	
	A permutation $\epi$ is said to be a \emph{palisade} if all components of $\spi$ are unoriented interleaving pairs. If $\spi$ consists of $\phi$ unoriented interleaving pairs, then we refer to $\epi$ as a palisade of size $\phi$, or a \emph{$\phi$-palisade}.
	
	\begin{example}
		The permutations $(0\;5\;4\;3\;2\;1\;6\;11\;10\;9\;8\;7\;12\;17\;16\;15\;14\;13)$ and $(0\;11\;16\;15\;14\;13\;12\;17\;10\;9\;8\;1\;6\;5\;4\;3\;2\;7)$ are examples of palisades. Their cycle graphs~\cite{BafnaPevzner1998} are depicted in Appendix~\ref{appendix-a} (figures~\ref{fig:palisade-1} and~\ref{fig:palisade-2}, respectively).
	\end{example}
	
	The next lemma can be demonstrated by observing that, for clearing two and three unoriented interleaving pairs, respectively, in the analysis carried out to prove the algorithms of Elias and Hartman~\cite{EliasHartman2006} and of Silva et al. \cite{silva2022new}, no $(4,3)$- and no $(8,6)$-sequence were found.
	
	\begin{lemma}[Elias and Hartman~\cite{EliasHartman2006}, Silva et al.~\cite{silva2022new}]\label{no_4_3}
		There is no $(4,3)$- and $(8,6)$-sequence of transpositions able to clear two and three unoriented interleaving pairs on a palisade, respectively.
	\end{lemma}
	
	We can use the lower bound based on the counting of hurdles by Christie~\cite{Christie1998} to establish a tighter transposition distance lower bound for palisades, setting the initial conditions required for determining the transposition distance of palisades. A \emph{hurdle} is a component of $\spi$ consisting only of odd-length cycles that cannot be cleared using only $2$-moves. It is easy to see that the number of hurdles in a $\phi$-palisade is $\phi$. Let $\epi$ be a $\phi$-palisade, then
	
	\begin{lemma}[Christie~\cite{Christie1998}]\label{lb}
		$d_t(\epi) \geq 2\phi + \lceil\frac{\phi}{2}\rceil$,
	\end{lemma}
	
	where $2\phi$ is $3$-norm of $\spi$ and $\lceil\frac{\phi}{2}\rceil$ is the minimum number of $0$-moves required to sort $\epi$ (i.e. one $0$-move for each pair of unoriented interleaving pairs). 
	
	\begin{theorem}
		\label{distance}
		If $\epi$ is $\phi$-palisade, then $d_t(\epi)=\Bigl\lceil\frac{11\phi}{4}\Bigr\rceil=\Bigl\lceil\frac{11(n+1)}{24}\Bigr\rceil$.
	\end{theorem}
	\begin{proof}
		Let us consider the lower bound established by Lemma~\ref{lb}. To sort $\epi$ using as few as $2\phi + \lceil\frac{\phi}{2}\rceil$ moves, it should be possible to clear two unoriented interleaving pairs using only one $0$-move (i.e., using a $(5,4)$-sequence, which includes a $(4,3)$-sequence). However, this is not possible due to Lemma~\ref{no_4_3}. Therefore, the lower bound can be raised to $2\phi + 2\lfloor\frac{\phi}{3}\rfloor + \phi \bmod 3$, meaning that at least two $0$-moves are required to clear three unoriented interleaving pairs --- the remaining ones require one $0$-move each. Nonetheless, this is not possible either, since Lemma~\ref{no_4_3} also states that there is no $(8,6)$-sequence to apply on $\epi$ clearing three unoriented interleaving pairs. As a result, the lower bound of $\phi$ can be raised again, using similar reasoning, to $2\phi + 3\lfloor\frac{\phi}{4}\rfloor + \phi \bmod 4$. Luckily, in this case, $(11,8)$-sequences~\cite{EliasHartman2006,silva2022new} can be optimally applied clearing four pairs at once using three $0$-moves. Based on this, we can conclude that $d_t(\epi)=2\phi + 3\lfloor\frac{\phi}{4}\rfloor + \phi \bmod 4$. Now, it remains to show the equality:
		\equationc{
			2\phi + 3\Bigl\lfloor\frac{\phi}{4}\Bigr\rfloor + \phi \bmod 4=\Bigl\lceil\frac{11\phi}{4}\Bigr\rceil.\label{eqdist}
		}
		
		For this, we will examine the congruence of $\phi$ modulo $4$, as follows.
		\begin{enumerate}
			\item $\phi \equiv 0 \pmod 4$. In this case, $\phi = 4m$ for some integer $m$. Thus,
			\alignc{
				2\phi + 3\Bigl\lfloor\frac{\phi}{4}\Bigr\rfloor + \phi \bmod 4 &= 2(4m) + 3\Bigl\lfloor\frac{4m}{4}\Bigr\rfloor \\ &= 11m = \Bigl\lceil\frac{44m}{4}\Bigr\rceil = \Bigl\lceil\frac{11\phi}{4}\Bigr\rceil.
			}
			\item $\phi \equiv 1 \pmod 4$. In this case, $\phi = 4m+1$ for some integer $m$. Thus,
			\alignc{
				2\phi + 3\Bigl\lfloor\frac{\phi}{4}\Bigr\rfloor + \phi \bmod 4 &= 2(4m+1) + 3\Bigl\lfloor\frac{4m+1}{4}\Bigr\rfloor + 1 \\
				&= 11m+3 \\ &= \Bigl\lceil\frac{44m+11}{4}\Bigr\rceil \\ &= \Bigl\lceil\frac{11(4m+1)}{4}\Bigr\rceil \\ &= \Bigl\lceil\frac{11\phi}{4}\Bigr\rceil.
			}
			\item $\phi \equiv 2 \pmod 4$. In this case, $\phi = 4m+2$ for some integer $m$. Thus,
			\alignc{
				2\phi + 3\Bigl\lfloor\frac{\phi}{4}\Bigr\rfloor + \phi \bmod 4 &= 2(4m+2) + 3\Bigl\lfloor\frac{4m+2}{4}\Bigr\rfloor + 2 \\
				&= 11m+6 \\ &= \Bigl\lceil\frac{44m+22}{4}\Bigr\rceil \\ &= \Bigl\lceil\frac{11(4m+2)}{4}\Bigr\rceil \\ &= \Bigl\lceil\frac{11\phi}{4}\Bigr\rceil
			}
			\item $\phi \equiv 3 \pmod 4$. Finally, we have $\phi = 4m+3$ for some integer $m$. Thus,
			\alignc{
				2\phi + 3\Bigl\lfloor\frac{\phi}{4}\Bigr\rfloor + \phi \bmod 4 &= 2(4m+3) + 3\Bigl\lfloor\frac{4m+3}{4}\Bigr\rfloor + 3 \\
				&= 11m+9 \\ &= \Bigl\lceil\frac{44m+33}{4}\Bigr\rceil \\ &= \Bigl\lceil\frac{11(4m+3)}{4}\Bigr\rceil \\ &= \Bigl\lceil\frac{11\phi}{4}\Bigr\rceil.
			}
		\end{enumerate}
		Therefore, Equation~\ref{eqdist} holds for all possible values of $\phi \bmod 4$, and hence for all $\phi$. To complete the proof, remark that a $\phi$-palisade is a permutation on $(n+1)\equiv 0 \pmod6$ symbols. Thus, $\lceil\frac{11\phi}{4}\rceil=\lceil\frac{11(n+1)}{24}\rceil$.
	\end{proof}
	
	\section{Transposition diameter for $3$-permutations}
	
	In this section, Theorem~\ref{distance} is used to provide the exact value for $TD3(n)$.
	
	Silva et al. \cite{silva2022new} introduced a new upper bound for the transposition distance, valid for any $\epi$, which can be stated as follows.
	
	\begin{theorem}[Silva et. al~\cite{silva2022new}]\label{upper-bound-distance}
		$d_t(\epi) \leq 11 \Bigl\lfloor \frac{\norm{\spi}_3}{8} \Bigr\rfloor + \Bigl\lfloor\frac{3(\norm{\spi}_3 \bmod 8)}{2}\Bigr\rfloor$.
	\end{theorem}
	
	Using this upper bound, we can devise a new upper bound for the transposition diameter for $3$-permutations (it should be noted that a $3$-permutation is a permutation on $(n+1)\equiv 0 \pmod3$ symbols).
	
	\begin{corollary}\label{upper-bound-td3}
		$TD3(n) \leq 11 \Bigl\lfloor \frac{n+1}{24} \Bigr\rfloor + \Bigl\lfloor\frac{(n+1) \bmod 24}{2}\Bigr\rfloor$.
	\end{corollary}
	
	Let $\epi$ be a $3$-permutation. We say $\epi$ is \emph{diametral} if $d_t(\epi)$ is equal to $TD3(n)$. One way to show that $\epi$ is diametral is by showing that $d_t(\epi)$ is equal to the upper bound of the transposition diameter stated above. This observation is useful for proving the next theorem, which provides the exact value for $TD3(n)$.
	
	\begin{theorem}
		\footnotesize{
			\begin{numcases}{TD3(n)=}
				\Bigl\lceil\frac{11(n+1)}{24}\Bigr\rceil, (n+1)\equiv 0\pmod6 \\
				\Bigl\lceil\frac{11(n-2)}{24}\Bigr\rceil + 1, (n+1)\equiv 3\pmod6
			\end{numcases}
		}
	\end{theorem}
	\begin{proof}
		We divide the proof into two parts.
		
		\begin{enumerate}
			\item If $(n+1)\equiv 0\pmod6$, then the palisades on $n+1$ symbols are diametral. The proof is by showing that their transposition distances are equal to the upper bound established by Corollary~\ref{upper-bound-td3}. That is, we have to show that
			\equationc{
				\Bigl\lceil\frac{11(n+1)}{24}\Bigr\rceil=11 \Bigl\lfloor \frac{n+1}{24} \Bigr\rfloor + \Bigl\lfloor\frac{(n+1) \bmod 24}{2}\Bigr\rfloor\label{eq1}.
			}
			We begin by noting that the fact $(n+1)\equiv 0\pmod6$ implies $n+1=6m$, for some $m$ integer. By replacing $n+1$ with $6m$ in the left-hand side of Equation~\ref{eq1}, we obtain: \[\lceil\frac{11(n+1)}{24}\rceil = \lceil\frac{11m}{4}\rceil.\] The same substitution made on the right-hand side yields:
			\alignc{
				11 \Bigl\lfloor \frac{n+1}{24} \Bigr\rfloor + \Bigl\lfloor\frac{(n+1) \bmod 24}{2}\Bigr\rfloor &= 11 \Bigl\lfloor \frac{6m}{24} \Bigr\rfloor + \Bigl\lfloor\frac{6m - 24\Bigl\lfloor\frac{6m}{24}\Bigr\rfloor}{2}\Bigr\rfloor \\
				&= 11 \Bigl\lfloor \frac{m}{4} \Bigr\rfloor + \Bigl\lfloor3m - 12\Bigl\lfloor\frac{m}{4}\Bigr\rfloor\Bigr\rfloor \\
				&= 11 \Bigl\lfloor \frac{m}{4} \Bigr\rfloor + 3m - 12\Bigl\lfloor\frac{m}{4}\Bigr\rfloor \\ &= 3m - \Bigl\lfloor\frac{m}{4}\Bigr\rfloor.
			}
			Now it remains to show that the equality below holds:
			\equationc{
				3m - \Bigl\lfloor\frac{m}{4}\Bigr\rfloor=\Bigl\lceil\frac{11m}{4}\Bigr\rceil\label{eq2}.
			}
			We will divide the analysis into two cases, next.
			\begin{enumerate}
				\item $m\equiv 0\pmod 4$. In this case, $m=4l$, for some integer $l$. Replacing $m$ with $4l$ in Equation~\ref{eq2} gives:
				\alignc{
					12l - \Bigl\lfloor\frac{4l}{4}\Bigr\rfloor=\Bigl\lceil\frac{11(4l)}{4}\Bigr\rceil=11l.
				}
				\item $m \not\equiv 0\pmod 4$. In this case, $m = 4l + r$ for some integer $l$ and $r \in [1, 3]$. By rewriting the left-hand side of Equation~\ref{eq2}, we have: \smallcenter{3m - \Bigl\lfloor\frac{m}{4}\Bigr\rfloor = 3(4l + r) - l = 11l + 3r.} On the right-hand side, \smallcenter{\Bigl\lceil\frac{11m}{4}\Bigr\rceil = \Bigl\lceil 11l + \frac{11r}{4}\Bigr\rceil=11l + \Bigl\lceil\frac{11r}{4}\Bigr\rceil.}
				Finally, note that $\lceil\frac{11r}{4}\rceil=3r$ for all $r \in [1, 3]$.
			\end{enumerate}
			
			\item Now we consider the case $(n+1)\equiv 3\pmod6$. In this scenario, we have identified another family of $3$-permutations that consists of $\frac{(n-2)}{6}$ unoriented interleaving pairs and an oriented $3$-cycle not intersecting with any other cycle of $\spi$. The transposition distances of these permutations are $\lceil\frac{11(n-2)}{24}\rceil + 1$, and we can prove that they are diametral (Figure~\ref{fig:diametral} shows the cycle graph~\cite{BafnaPevzner1998} of such a permutation). For this, we have to demonstrate the following equality:
			\equationc{
				\Bigl\lceil\frac{11(n-2)}{24}\Bigr\rceil + 1=11 \Bigl\lfloor \frac{n+1}{24} \Bigr\rfloor + \Bigl\lfloor\frac{(n+1) \bmod 24}{2}\Bigr\rfloor\label{eq3}.
			}
			Note that if $(n+1) \equiv 3 \pmod 6$, then $n+1=6m+3$ for some integer $m \geq 1$. By substituting this expression for $n+1$ into the left-hand side of Equation~\ref{eq3}, we obtain:
			\alignc{
				\Bigl\lceil\frac{11(n-2)}{24}\Bigr\rceil + 1 &= \Bigl\lceil\frac{11m}{4}\Bigr\rceil + 1.
			}
			The same substitution on right-hand side gives:
			\alignc{
				11 \Bfloor{\frac{n+1}{24}} + \Bfloor{\frac{(n+1) \bmod 24}{2}} &= 11 \Bfloor{\frac{6m+3}{24}} + \Bfloor{\frac{6m+3-24\Bfloor{\frac{6m+3}{24}}}{2}}\\
				&= 11 \Bfloor{\frac{6m+3}{24}} + \Bfloor{\frac{6m+3}{2}} - 12\Bfloor{\frac{6m+3}{24}}\\
				&= \Bfloor{\frac{6m+3}{2}} - \Bfloor{\frac{6m+3}{24}}\\
				&= \Bfloor{3m+\frac{3}{2}} - \Bfloor{\frac{m}{4}+\frac{3}{24}}\\
				&= 3m - \Bfloor{\frac{m}{4}+\frac{3}{24}} + 1
			}
			Our focus now is to show that:
			\equationc{
				3m - \Bfloor{\frac{m}{4}+\frac{3}{24}} + 1 = \Bigl\lceil\frac{11m}{4}\Bigr\rceil + 1 .\label{eq4}
			}
			For this, let us analyse the congruence of $m$ modulo $4$.
			\begin{enumerate}
				\item $m\equiv 0\pmod 4$. In this case, $m=4l$, for some integer $l$. By rewriting the terms, we have:
				\alignc{
					3m - \Bfloor{\frac{m}{4}+\frac{3}{24}} + 1 &= 12l + 1 - \Bfloor{l+\frac{3}{24}} \\ &= 11l + 1 \\ &= \Bigl\lceil\frac{11m}{4}\Bigr\rceil + 1.
				}
				Therefore, Equation~\ref{eq4} holds.
				\item $m \not\equiv 0\pmod 4$. In this case, $m = 4l + r$ for some integer $l$ and $r \in [1, 3]$. Rewriting the left-hand side of Equation~\ref{eq4} gives:
				\alignc{
					\Bceil{\frac{11m}{4}} + 1 &=\Bceil{\frac{11(4l+r)}{4}} + 1 \\ &=\Bceil{11l+\frac{11r}{4}} + 1 \\ &= 11l+\Bceil{\frac{11r}{4}} + 1.
				}
				On the right-hand side, we obtain:
				\alignc{
					3m - \Bfloor{\frac{m}{4}+\frac{3}{24}} + 1 &= 3(4l+r) - \Bfloor{\frac{4l+r}{4}+\frac{3}{24}} + 1\\
					&= 12l+3r - \Bfloor{l+\frac{r}{4}+\frac{3}{24}} + 1\\
					&= 12l+3r - (l + \Bfloor{\frac{r}{4}+\frac{3}{24}}) + 1 \\ &= 11l+3r + 1.
				}
				Note that the term $\lfloor\frac{r}{4}+\frac{3}{24}\rfloor$ in the expression above was cancelled since it always evaluates to $0$. The observation that $3r=\lceil\frac{11r}{4}\rceil$ holds for all $r \in [1, 3]$ concludes the proof.
			\end{enumerate}
		\end{enumerate}
	\end{proof}
	
	\section{A barrier for approximating SBT}
	
	We observe that, when sorting a permutation $\epi$ by applying sequences of transpositions obtained by a polynomial-time routine, palisades of any sizes $\phi$, $0 < 2\phi \leq \norm{\spi}_3$, may appear. The lowest achievable approximation ratio when sorting a palisade --- using Bafna and Pevzner's~\cite{BafnaPevzner1998} as lower bound, is \smallcenter{ \bigg(\Bigl\lceil\frac{11(n+1)}{24}\Bigr\rceil \mathbin{/} \frac{n+1}{3}\bigg) \geq \bigg( \frac{33}{24} = 1.375\bigg), } where $\frac{n+1}{3}$ is the $3$-norm of a palisade on $n+1$ symbols. This remark has a striking implication that the palisades are an obstacle to approximating SBT with approximation ratios lower than $1.375$.
	
	\begin{corollary}\label{cor}
		It is impossible to guarantee an approximation ratio lower than $1.375$ when sorting by transpositions using as lower bound, the one devised by Bafna and Pevzner~\cite{BafnaPevzner1998}.
	\end{corollary}
	
	To illustrate Corollary~\ref{cor}, we provide the following example: an arbitrary permutation in $S_n$ on which we apply a sequence of $2$-moves resulting in a palisade. In this example, it is worth noting that, based only on the most apparent properties of the cycles of $\spi$ (such as cycle parity), we cannot set a lower bound for $d_t(\epi)$ that is tighter than the one proposed by Bafna and Pevzner~\cite{BafnaPevzner1998}.
	
	\begin{example}
		Consider $\epi=(0\;23\;22\;21\;1\;6\;5\;11\;20\;10\;9\;8\;13\;4\;3\;7\;12\;18\;2\;17\;16$ $\;15\;14\;19)$ (cycle graph~\cite{BafnaPevzner1998} depicted in Figure~\ref{fig:arbitrary}). We can apply on $\epi$ a sequence of three $2$-moves, namely: $\tau_1=(1\;20\;10)$, $\tau_2=(2\;17\;7)$ and $\tau_3=(4\;18\;11)$, giving $\tau_3\tau_2\tau_1\epi=(0\;23\;22\;21\;20\;1\;6\;5\;4\;3\;2\;7\;12\;11\;10\;9\;8\;13\;18\;17\;16\;15\;14\;19)$, which is a $4$-palisade.
	\end{example}

	
	We highlight that attempting to determine in advance whether a sequence of transpositions will result in a palisade would be impractical. In this regard, we recall that,  even for a $3$-permutation, deciding whether or not a single $0$-move is required for a sorting, is $\mathcal{NP}$-hard~\cite{bulteau2012sorting}.
	
	\section{Conclusions}
	
	This study focused on palisades, a family of permutations that are particularly ``hard'' to sort. We demonstrated that for $(n+1) \equiv 0\pmod 6$, palisades qualify as diametral permutations for the $3$-permutations subset of $S_n$. Moreover, we have shown that there exists another family of diametral $3$-permutations for $(n+1) \equiv 3\pmod 6$ (a $3$-permutation is a permutation on $(n+1) \equiv 0\pmod 3$). Elias and Hartman~\cite{EliasHartman2006} had previously established an upper bound for the transposition diameter for $3$-permutations. In this work, we have provided the exact value.
	
	Interestingly, our findings also indicate that palisades pose an obstacle to achieving approximation ratios lower than $1.375$ for SBT. This partly answers a question raised by Bulteau, Fertin and Rusu~\cite{bulteau2012sorting} about the feasibility of devising a polynomial-time scheme for approximating SBT. Such a scheme is not possible using as lower bound, the well-known one by Bafna and Pevzner~\cite{BafnaPevzner1998}. Therefore, for the study of approximate solutions for SBT to advance, it is crucial to find new lower bounds for the transposition distance.
	
	\appendix
	
	\section{Cycle graph}\label{appendix-a}
	
	The cycle graph is a commonly used graphical representation for permutations in the context of SBT~\cite{BafnaPevzner1998}. To construct the cycle graph of $\pi=[\pi_1\;\pi_2\dots \pi_n]$, two extra symbols $\pi_0=0$ and $\pi_{n+1}=n+1$ are added to it. Thus, \emph{cycle graph} of $\pi$, denoted by $G(\pi)$, is a directed graph consisting of the set of vertices $\{+0,\;-1,\;+1,\;-2,\;+2,\;\dots,$ $\;-n,\;+n,\;-(n+1)\}$ and a set of coloured edges (which can be black or gray). The black edges connect $-\pi_{i}$ to $+\pi_{i-1}$ for all $i\in[1,n+1]$. The gray edges connect vertex $+i$ to vertex $-(i+1)$ for all $i\in[0,n]$. The black edges represent the current state of genes in the chromosome of $\pi$, while the gray edges indicate the desired order in the second permutation $\iota=[1\;2\dots n]$. Figures~\ref{fig:palisade-1}, ~\ref{fig:palisade-2} and ~\ref{fig:diametral} provides examples of cycle graphs.
	
	Each vertex in $G(\pi)$ has in-degree and out-degree equal to $1$, with one black edge entering and one gray edge leaving. This induces a unique cycle decomposition in $G(\pi)$. A \emph{$\kappa$-cycle} is a cycle $C$ in $G(\pi)$ with $\kappa$ black edges. If $\kappa$ is \emph{even (odd)}, we also say that $C$ is an \emph{even} (\emph{odd}) cycle.
	
	A transposition $\tau(i,j,k)$, with $1\leq i<j<k\leq n+1$, exchanges the block of symbols corresponding to interval $[\pi_i,\pi_{j-1}]$ with the block of symbols corresponding to interval $[\pi_j,\pi_{k-1}]$. Thus, applying $\tau(i,j,k)$ on $\pi$, denoted by $\tau(i,j,k)\cdot\pi$ or just $\tau\cdot\pi$, yields the new permutation $[\pi_1\dots\pi_{i-1}\;\pi_j\dots\pi_{k-1}$ $\;\pi_i\dots\pi_{j-1}\;\pi_k\dots\pi_n]$. Let $c_{odd}(\pi)$ be the number of odd cycles in $G(\pi)$ and let $\Delta c_{odd}(\pi,\tau)=c_{odd}(\tau\cdot\pi)-c_{odd}(\pi)$ denote the change in the number of odd cycles in $G(\pi)$ after the application of a transposition $\tau$. The following lower bound, formulated by Bafna and Pevzner~\cite{BafnaPevzner1998}, is the foundation of the most prevalent approximation algorithms for SBT.
	
	\begin{theorem}[Bafna and Pevzner~\cite{BafnaPevzner1998}]
		\label{th:bafna}
		$d(\pi)\geq\frac{n+1-c_{odd}(\pi)}{2}$.
	\end{theorem}
	
	\begin{figure*}[ht]
		\centering
		\includegraphics[width=\textwidth]{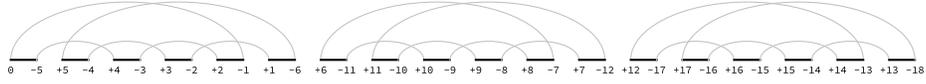}
		\caption{Cycle graph of the $3$-palisade $[5\;4\;3\;2\;1\;6\;11\;10\;9\;8\;7\;12\;17\;16\;15\;14\;$ $13]$.}
		\label{fig:palisade-1}
	\end{figure*}
	
	\begin{figure*}[ht]
		\centering
		\includegraphics[width=\textwidth]{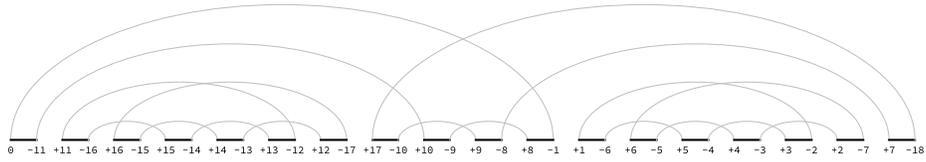}
		\caption{Cycle graph of the $3$-palisade $[11\;16\;15\;14\;13\;12\;17\;10\;9\;8\;1\;6\;5\;4\;3\;2$ $\;7]$.}
		\label{fig:palisade-2}
	\end{figure*}
	
	\begin{figure*}[ht]
		\centering
		\includegraphics[width=\textwidth]{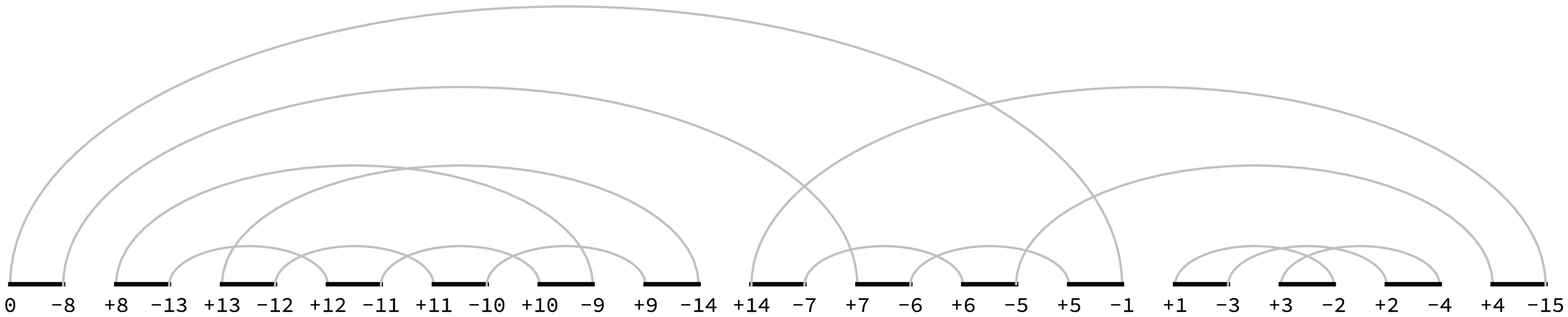}
		\caption{Cycle graph of the diametral $3$-permutation $[8\;13\;12\;11\;10\;9\;14\;7\;6\;5\;$ $1\;3\;2\;4]$.}
		\label{fig:diametral}
	\end{figure*}
	
	\begin{figure*}[ht]
		\centering
		\includegraphics[width=1.0\textwidth]{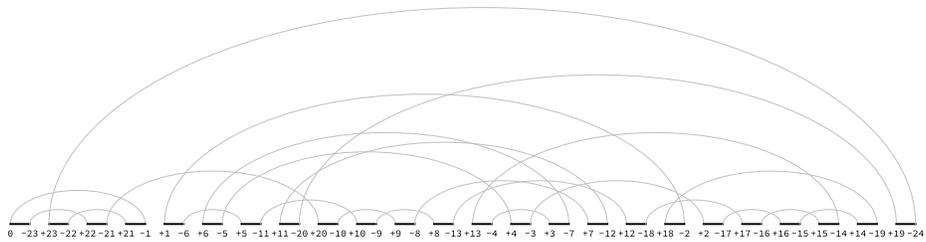}
		\caption{Cycle graph of $[23\;22\;21\;1\;6\;5\;11\;20\;10\;9\;8\;13\;4\;3\;7\;12\;18\;2\;17\;16\;15\;$ $14\;19]$.}
		\label{fig:arbitrary}
	\end{figure*}
	
    \clearpage
	
	\bibliographystyle{splncs04} 
	\bibliography{manuscript}

\begin{thebibliography}{10}
\providecommand{\url}[1]{\texttt{#1}}
\providecommand{\urlprefix}{URL }
\providecommand{\doi}[1]{https://doi.org/#1}

\bibitem{alexandrino20221}
Alexandrino, A.O., Brito, K.L., Oliveira, A.R., Dias, U., Dias, Z.: A
  1.375-approximation algorithm for sorting by transpositions with faster
  running time. In: Advances in Bioinformatics and Computational Biology: 15th
  Brazilian Symposium on Bioinformatics, BSB 2022, Buzios, Brazil, September
  21--23, 2022, Proceedings. pp. 147--157. Springer (2022)

\bibitem{BafnaPevzner1998}
Bafna, V., Pevzner, P.A.: {Sorting by transpositions}. SIAM Journal on Discrete
  Mathematics  \textbf{11}(2),  224--240 (1998)

\bibitem{bulteau2012sorting}
Bulteau, L., Fertin, G., Rusu, I.: Sorting by transpositions is difficult. SIAM
  Journal on Discrete Mathematics  \textbf{26}(3),  1148--1180 (2012)

\bibitem{Christie1998}
Christie, D.A.: Genome Rearrangement Problems. Ph.D. thesis, Glasgow University
  (1998)

\bibitem{dummit2004abstract}
Dummit, D.S., Foote, R.M.: Abstract algebra, vol.~3. Wiley Hoboken (2004)

\bibitem{EliasHartman2006}
Elias, I., Hartman, T.: A 1.375-approximation algorithm for sorting by
  transpositions. IEEE/ACM Transactions on Computational Biology and
  Bioinformatics  \textbf{3}(4),  369--379 (2006)

\bibitem{Gallian2009}
Gallian, J.: Contemporary Abstract Algebra. Brooks Cole, Boston, MA, 7 edn.
  (Jan 2009)

\bibitem{hartman2006simpler}
Hartman, T., Shamir, R.: A simpler and faster 1.5-approximation algorithm for
  sorting by transpositions. Information and Computation  \textbf{204}(2),
  275--290 (2006)

\bibitem{Hartman2005}
Hartman, T., Sharan, R.: {A 1.5-approximation algorithm for sorting by
  transpositions and transreversals}. J. Comput. Syst. Sci.  \textbf{70}(3),
  300--320 (2005)

\bibitem{koonin2005orthologs}
Koonin, E.V.: Orthologs, paralogs, and evolutionary genomics. Annual Review of
  Genetics  \textbf{39},  309--338 (2005)

\bibitem{MeidanisDias2000}
Meidanis, J., Dias, Z.: An Alternative Algebraic Formalism for Genome
  Rearrangements, pp. 213--223. Springer Netherlands, Dordrecht (2000)

\bibitem{mira2005algebraic}
Mira, C.V.G., Meidanis, J.: Algebraic formalism for genome rearrangements.
  Technical Report, Institute of Computing, University of Campinas  (June 2005)

\bibitem{NadeauTaylor1984}
Nadeau, J.H., Taylor, B.A.: Lengths of chromosomal segments conserved since
  divergence of man and mouse. Proc. Natl. Acad. Sci. USA  \textbf{81}(3),
  814--818 (1984)

\bibitem{PalmerHerbon1988}
Palmer, J.D., Herbon, L.A.: Plant mitochondrial dna evolves rapidly in
  structure, but slowly in sequence. Journal of Molecular Evolution
  \textbf{28},  87--97 (1988)

\bibitem{silva2022new}
Silva, L.A.G., Kowada, L.A.B., Rocco, N.R., Walter, M.E.M.: A new
  1.375-approximation algorithm for sorting by transpositions. Algorithms for
  Molecular Biology  \textbf{17}(1),  1--17 (2022)

\bibitem{YuePhylogenetics}
Yue, F., Zhang, M., Tang, J.: Phylogenetic reconstruction from transpositions.
  BMC Genomics  \textbf{9}(S15),
  https://doi.org/10.1186/1471--2164--9--S2--S15 (2008)

\end{thebibliography}
\end{document}